\documentclass[a4paper,onecolumn,11pt]{quantumarticle}
\pdfoutput=1
\usepackage[utf8]{inputenc}
\usepackage[english]{babel}
\usepackage[T1]{fontenc}
\usepackage{amsmath,amssymb,amsfonts,mathtools}
\usepackage{enumerate}
\usepackage{mathtools}
\usepackage{hyperref}

\usepackage[backend=bibtex]{biblatex}
\usepackage{tikz}
\usepackage{lipsum}

\newtheorem{theorem}{Theorem}
\newtheorem{definition}{Definition}
\newtheorem{lemma}{Lemma}
\newtheorem{axiom}{Axiom}

\newenvironment{proof}
  {\par\noindent\textit{Proof.}\enspace} % Formatting for the start (e.g., italicized "Proof.")
  {\hfill$\square$\par\medskip}

\def\H{\ensuremath\mathcal{H}}

\def\C{\ensuremath\mathbb{C}}
\def\R{\ensuremath\mathbb{R}}
\def\Z{\ensuremath\mathbb{Z}}
\def\N{\ensuremath\mathbb{N}}

\def\Q{\ensuremath\mathcal{Q}}

\newcommand{\inner}[2]{\langle{}#1,#2\rangle}

\newcommand{\bra}[1]{\langle{}#1|}
\newcommand{\ket}[1]{|#1\rangle}

\newcommand{\proj}[1]{\ket{#1}\bra{#1}}

\DeclareMathOperator*{\im}{Im}
\DeclareMathOperator*{\tr}{Tr}

\DeclareMathOperator*{\rank}{rank}
\DeclareMathOperator*{\nullspace}{null}
\DeclareMathOperator*{\stab}{Stab}

\newcommand{\mattwo}[4]{
  \begin{bmatrix}
    #1 & #2\\
    #3 & #4
  \end{bmatrix}
}

\begin{document}

\title{Lower Bound for Region of Non-Distillable Magic States}

\author{Akul Sethi}
\affiliation{Department of Mathematics, Northeastern University, Boston, Massachusetts}
\maketitle

\begin{abstract}
  Magic states are a leading proposal for realizing fault-tolerant and universal quantum computation. At the present moment, much of the theory of magic states is not understood. Even simple questions such as what single qubit states have magic remain open. In Campbell and Browne's seminal paper on bound magic states, they prove the existence of states outside of the stabilizer octahedron which cannot be used to achieve universal computation. In this paper we expand upon their proof and determine a lower-bound on this region.  
\end{abstract}

\tableofcontents{}
\newpage
\section{Introduction}

To physically realize useful quantum computation, two critical components must be realized: universal computation and error-correction. The interplay between the two is what leads to the study of quantum magic \cite{BravyiKitaev2005}. Universal computation is the ability of a quantum device to approximate any quantum algorithm up to arbitrary precision while quantum error-correction is collection of algorithms and techniques to make quantum calculation robust to environmental noise.

Within the formalism of quantum mechanics, a quantum system can be represented as a Hilbert space $\mathcal{H}$. In the case of computation, we can generally write $H = H_1 \otimes \dots \otimes H_k$ where each $H_1$ is a 2-dimensional Hilbert space known as a qubit. A highly desirable property of an error-correction scheme is that any operation on the space, known as gates, can be split into independent operations on the subsystems. Such a collection of operations is known as a transversal gate set and allows errors on the qubits to remain uncorrelated. It is then a natural question to ask if there exists a gateset which is both universal and transversal. The Eastin-Knill theorem illustrates that there is not \cite{EastinKnill2009}. The theory of quantum magic was developed to circumvent this. 

It is a well known fact that the Clifford Group along with the $T$ gate forms a universal gate set for quantum computation \cite{NielsenChuang2010}. Moreover, it can be shown that the entire Clifford Group can be implemented transversally \cite{Steane1997}. Combining these, in 2005, Bravyi and Kitaev were able to illustrate that given access to a supply of specific states they coin "magic states", they were able to implement the T gate transversally as well \cite{BravyiKitaev2005}. Moreover, they present an algorithm which takes imperfect copies of magic sates and "distills" them towards true magic states using only Clifford operations. Note that this is not a violation of the Eastin-Knill theorem as that theorem does not assume access to any ancilla states. In order to not violate the theorem, it is necessary that the construction of such states involves non-transversal operations. 

Another natural question now is what states can be distilled towards magic ones. Gottesman and Knill were able to show that there is a region of states, known as stabilizer states (which will be defined precisely later on) which cannot be distilled \cite{Gottesman1998Heisenberg}. Later, Campbell and Browne were able to show that this region extends past the stabilizer states, but do so without defining what this region is \cite{CampbellBrowne2010}. In this project we develop lower bounds for this region. 

\section{Background}

\subsection{Introduction to Quantum Mechanics}
We will begin with a brief axiomatization of quantum mechanics in a way which is useful for quantum information theory \cite{NielsenChuang2010}. 
\begin{definition}[Hilbert Space]
  A Hilbert space $\H$ is a finite dimensional vector space over $\C$ equipped with inner product $\inner{-}{-} : \H \times \H \rightarrow \C$ such that for all $\alpha,\beta \in C$ and $x,y,z \in \H$
  \begin{enumerate}
    \item $\inner{z}{\alpha x + \beta y} = \alpha\inner{z}{x} + \beta\inner{z}{y}$ 
    \item $\inner{x}{y} = \overline{\inner{y}{x}}$ 
    \item $\inner{x}{x} \geq 0$
  \end{enumerate}
\end{definition}
In this paper we will use braket notation. For any $\psi \in \H$ we will define 
\begin{equation}
  \ket{\psi} \coloneqq \psi \quad \text{ and } \quad \bra{\psi} \coloneqq (\phi \mapsto \inner{\psi}{\phi})
\end{equation}
where $\bra{\psi} : \H \rightarrow \C$ is a linear map. 
\begin{axiom}[State Axiom]
  For any physical quantum system we can associate a Hilbert space which can be used to define its state. Precisely, let $\H$ be a Hilbert space. A linear map $\rho : \H \rightarrow \H$ is a state if for all $\psi \in \H$
  \begin{enumerate}
    \item $\tr{(\rho)} = 1$ 
    \item $\rho = \rho^\dagger$
    \item $\bra{\psi}\rho\ket{\psi} \geq 0$
  \end{enumerate}
\end{axiom}
If $\rho = \ket\psi\bra\psi$ for some unit norm $\ket\psi \in \H$ it is known as a pure state. In this case we will just use $\ket{\psi}$ to denote the state. Otherwise, $\rho$ is known as a mixed state. It is easy to show that all mixed states can be written as convex combinations of pure ones. The term state will be used to refer to $\rho$ of either form where the type should be easy to determine from context. We will use the term n-qubit state to refer to a state with $\dim{\H} = 2^n$. Since it is more natural in quantum information theory to think of systems in terms of their number of qubits, we define $\Q_n := \C^{2^n}$ and $\mathcal{S}(\Q_n)$ to be the set of states $\rho: \Q_n \rightarrow \Q_n$. 
\begin{axiom}
  The evolution of a closed quantum system is described by a unitary transformation. That is, the state $\rho$ of the system at time $t_1$ is related to the state $\rho'$ are time $t_2$ by unitary $U$ which only depends on $t_1$ and $t_2$,
  \begin{equation}
    \rho' = U\rho U'
  \end{equation}
\end{axiom}
\begin{axiom}
  Quantum measurements are described by a collection $\{M_m\}$ where $M_m : \H \rightarrow \H$ are linear maps satisfying
  \begin{equation}
    \sum_m M_{m}^\dagger M_m = I
  \end{equation}
  The index $m$ refers to the possible outcomes of the measurement. The probability that measurement outcome $m$ takes place is given by 
  \begin{equation}
    p(m) = \tr{(M_{m}^\dagger M_m \rho)}
  \end{equation}
  and the state after measurement is
  \begin{equation}
    \frac{M_m \rho M_{m}^\dagger}{\tr{(M_{m}^\dagger M_m \rho)}}
  \end{equation}
\end{axiom}
\begin{axiom}
  The state space of a composite physical system is the tensor product of component physical systems. Moreover given $n$ systems where the ith one is in state $\rho_i$, then the joint state of the total system is $\rho_1 \otimes \rho_2 \otimes \dots \otimes \rho_n$.
\end{axiom}

\subsection{Pauli, Clifford, and Stabilizer Groups}
There are several important groups which are central in studying Quantum Information Theory. In this section we define them and mention useful properties which are used throughout.
\begin{definition}[Pauli Group]
  We will define the following 
  \begin{equation*}
    X = \mattwo{0}{1}{1}{0}, \quad Y = \mattwo{0}{i}{-i}{0}, \quad Z = \mattwo{1}{0}{0}{-1}
  \end{equation*}
  The Pauli group $P_n \leq \text{GL}_{2^n}(\C)$ is a subgroup where for all $P \in P_n$ we can write
  \begin{equation}
    P = \eta P_1 \otimes \dots \otimes P_n 
  \end{equation}
  for some $\eta \in \{\pm 1 \pm i\}$ and $P_i \in \{I,X,Y,Z\}$.
\end{definition}
We will refer to $\eta$ as the phase of an operator in $P_n$. Sometimes, we only want to consider phase-less operators and so we will use $\overline{P_n} \leq P_n$ to be the subgroup with $\eta = 1$. Note that each of $X,Y,Z$ are traceless, hermitian, unitary, and involutionary. Since they show up repeatedly, we will define $R := \{X,Y,Z\}$ and use $R_\pm$ to denote the same set with phases $\pm$. Moreover, all elements of $P_n$ either commute or anti-commute. Together with $I$, $R$ forms an orthogonal basis for the set of $2 \times 2$ hermitian matrices. It is common to denote the $+1$ eigenstates of $X,Y,Z$ as $\ket{+},\ket{+i},\ket{0}$ respectively, and $\ket{-},\ket{-i},\ket{1}$ for the $-1$ eigenstates. 
\begin{definition}[Clifford Group]
  The Clifford group $\mathcal{C}_n$ is defined 
  \begin{equation}
    \mathcal{C}_n \coloneqq \{V \in U(2^n) : VP_nV^\dagger = P_n\} / U(1)
  \end{equation}
  where $U(n)$ is the unitary group of $n \times n$ complex matrices. 
\end{definition}
Finally, we will define a special class of group which is used throughout quantum error correction. 
\begin{definition}[Stabilizer Group]
Let $S \leq P_n$ be an Abelian subgroup. If $-I \notin S$ we say $S$ is a stabilizer subgroup. The simultaneous $+1$ eigenspace of $S$ is referred to as its stabilized subspace and denoted $\stab{(S)}$. 
\end{definition}
Notice that all stabilizer subgroups cannot have any elements with phase $\pm i$. Since all phase-less Pauli's are involutionary, by closure this would imply that the stabilizer group contains $-I$ which is a contradiction. Throughout this paper, for any stabilizer subgroup $S$, we will use $\Pi_S$ to denote the orthogonal projection onto $\stab{(S)}$.
\begin{definition}[Rank]
Let $S \leq P_n$ be a stabilizer group generated by some set $G = \{s_1, \dots, s_k\}$. If $\langle G \setminus \{s_i\}\rangle \neq S$ for all $s_i \in G$ we call $G$ an independent set of generators and call $k$ the rank of $S$.
\end{definition}
\begin{definition}[Stabilizer Normalizer]
  Let $S \leq P_n$ be a stabilizer subgroup. We define 
  \begin{equation}
  N(S) := \{P \in P_n : Ps = sP,\ \forall s \in S\}
  \end{equation}
\end{definition}
It would appear that we have defined the centralizer of $S$ rather than its normalizer but it turns out that for the Pauli group, they are the same. 
\begin{definition}[Stabilizer State]
  $\ket\psi \in \Q_n$ is a pure stabilizer state if for some rank $n$ stabilizer subgroup $S \leq P_n$, $\ket\psi \in \stab{(S)}$. 
\end{definition}
Like with states, convex combinations of pure stabilizer states are referred to as mixed stabilizer states while the term stabilizer state refers to a state of either form. In our paper we will use the following standard result. 

\begin{lemma}
  Let $S \leq P_n$ be some rank $k$ stabilizer subgroup. Then, $N(S) / S \cong P_{n - k}$
  \label{lemma:normalizer}
\end{lemma}
 A proof of this fact which uses the binary symplectic representation of the Pauli group is included in \cite{surviving_quantum_world}. More standard results on stabilizer groups, including the symplectic representation, are included in Appendix \ref{app:stab}. For the remainder of this paper we will assume that the reader has familiarized themselves with these facts.

\subsection{Bloch Sphere Representation}

We had mentioned earlier that the Pauli matrices along with $I$ gives an orthogonal basis for the space of $2 \times 2$ complex Hermitian matrices. It turns out that we can use this fact to get a nice geometric representation of single-qubit states. Moreover, this representation also gives us a simpler way to determine if a state is stabilizer. 

\begin{definition}[Bloch Sphere Representation, B]
  For any single qubit state $\rho$ there is a unique $B(\rho) \in \R^3$ with $||B(\rho)||_2 \leq 1$ such that  
  \begin{equation}
    \rho = \frac{I + B(\rho)_XX + B(\rho)_YY + B(\rho)_ZZ}{2}
  \end{equation}
  We call $B(\rho)$ the Bloch sphere representation of $\rho$. 
  \label{def:bloch}
\end{definition}
 Using this representation we have the following correspondences $\ket{+} \mapsto (1,0,0), \ket{-} \mapsto (-1,0,0), \ket{+i} \mapsto (0,1,0), \ket{-i} \mapsto (0,-1,0), \ket{0} \mapsto (0,0,1), \ket{1} \mapsto (0,0,-1)$ using the labelling $(x,y,z) \in \R^3$. Since $B$ respects convex combinations, the set of single qubit stabilizer states is the convex hull of these six vertices. Viewed geometrically, this forms an octahedron within the Bloch sphere. Therefore, a state is stabilizer if and only if the L1 norm of its Bloch vector is less than or equal to 1. Since our bounds will be expressed in terms of this quantity we give it a name. 
\begin{definition}[L1 Bloch Norm, F]
  Let $B(\rho)$ be the Bloch sphere representation of some single qubit state $\rho$. Then the L1 Bloch norm is defined,
  \begin{equation}
    F(\rho) = |B(\rho)_X| + |B(\rho)_Y| + |B(\rho)_Z|
  \end{equation}
\end{definition}
Since Clifford operators simply permute the Pauli matrices, the L1 Bloch norm of a state is preserved under Clifford operations. Note that we will drop the argument of $B$ and $F$ when the input state is obvious. 

\subsection{Stabilizer Codes}
The central tool used in proving the main result relies on the stabilizer formalism. Historically, this was developed in Gottesman's Ph.D. thesis to produce quantum error correcting codes, but in this paper we will be utilizing the formalism in a slightly different way. For this reason, some of the definitions may feel unmotivated, but this should not be an issue in understanding the results. Let $W \subseteq \Q_n$ be some 2 dimensional subspace. In the context of error-correction, this is called a codespace. 
\begin{definition}[Encoding and Decoding]
  An encoding map for $W$ is any isometry $V : \Q_1 \rightarrow \Q_n$ whose image is $W$. An encoding $V$ is said to be Clifford if there exists some $C \in \mathcal{C}_n$ such that 
  \begin{equation}
    V\ket\psi = C(\ket\psi \otimes \ket{0^{n - 1}})
  \end{equation}
  In this case we call $C^{-1} = C^\dagger$ a Clifford decoding for $V$.
\end{definition}

All encodings and decodings from here on will be considered Clifford, but we mention standard encodings for completeness. Note that if $C \in \mathcal{C}_n$ is a decoding for some encoding of $W$ then we call it a decoding for $S$. We are especially interested in the case where $S$ is the stabilized subspace of some rank $n - 1$ stabilizer subgroup. In this case, it actually turns out that the existence of Clifford encodings is guaranteed. To show this, we will use the following standard result without proof.

\begin{lemma}
  Let $\{s_1, \dots, s_n\}$ and $\{s_{1}', \dots, s_{n}'\}$ be any two sets of independent n-qubit Pauli operators, generating stabilizer groups $S$ and $S'$ respectively. Then there exists $C \in \mathcal{C}_n$ such that 
  \begin{equation}
    Cs_iC^\dagger = s_{i}'
  \end{equation}
  for all $i \in [n]$
  \label{lemma:cliff}
\end{lemma}
Using this, we can prove the following:

\begin{lemma}
  If $W = \stab{(S)}$ for some stabilizer subgroup $S \leq P_n$, then $W$ admits a Clifford encoding.
  \label{lemma:encoding_exists}
\end{lemma}
\begin{proof}
  Suppose that $W = \stab{(S)}$ for some rank $n - 1$ stabilizer subgroup $S$. Define the map $\ket\psi \mapsto \ket\psi \otimes \ket{0^{n - k}}$ which takes a single qubit state $\ket\psi \in \Q_1$ to an n-qubit state. The image of this map is a stabilizer code $S_0$ stabilized by generators $\{Z_2, \dots, Z_n\}$. Let $\{g_1, \dots, g_{n - 1}\}$ be independent generators of $S$. By Lemma \ref{lemma:cliff}, there exists $C \in \mathcal{C}_n$ such that $CZ_iC^\dagger = g_i$. Then, 
  \begin{equation}
    V\ket\psi := C(\ket\psi \otimes \ket{0^{n - 1}})
  \end{equation}
  is a Clifford encoding map $V : \Q_1 \rightarrow \Q_n$ whose image is $W$. 
\end{proof}

Since the codespace is simply a copy of $\Q_1$, we are interested in the operators in $P_n$ which mimic the action of $P_1$ on the codespace. These are known as the \textit{logical operators} of the code. Precisely, the set of logical operators is equivalent to $N(S)$. Similarly, we are interested in the $\textit{logical states}$ of the code with respect to some logical operator. If we let $\overline{Z} \in N(S) \setminus \pm S$ then there are unique $\ket{\overline{0}},\ket{\overline{1}} \in \stab(S)$ such that
\begin{equation}
  \overline{Z}\ket{\overline{0}} = \ket{\overline{0}}, \qquad  \overline{Z}\ket{\overline{1}} = -\ket{\overline{1}}
\end{equation}

\begin{lemma}
  Let $S \leq P_n$ be a rank $n - 1$ stabilizer subgroup with a choice of Z-logical $\bar{Z} \in N(S) \setminus \pm S$. Then there exists a Clifford encoding $V: \Q_1 \rightarrow \Q_n$ such that 
  \begin{equation}
    V\ket{z} = \ket{\bar{z}}
  \end{equation}
  for all $z \in \{0,1\}$. 
  \label{lemma:encoding_logical}
\end{lemma}
\begin{proof}
  Let $V_0 : \Q_1 \rightarrow \Q_n$ be a Clifford encoding which exists by lemma \ref{lemma:encoding_exists}. This encoding defines its own Z-logical $\bar{Z}'$. Specifically, 
  \begin{equation}
    \bar{Z}' = V_0ZV_{0}^\dagger
  \end{equation}
  Let $\{s_1,\dots,s_{n - 1}\}$ be an independent set of stabilizer generators for $S$. Since $\{s_1, \dots, s_{n - 1}, \bar{Z}'\}$ and $\{s_1,\dots,s_{n - 1},\bar{Z}\}$ are independent generators of stabilizer groups by lemma \ref{lemma:encoding_exists} there exists $C$ which maps between them such that
  \begin{equation}
    C\bar{Z}'C^\dagger = \bar{Z}
  \end{equation}
  If we define $V = CV_0$ it then follows that $V$ is a Clifford encoding which satisfies 
  \begin{equation}
    V\ket{z} = \ket{\bar{z}}
  \end{equation}
\end{proof}

\subsection{Reductions}

In magic state distillation, we are interested in protocols which have access to a restricted set of operations rather then the full set allowable under the axioms of quantum mechanics. In particular we are interested in a class of protocols we call Clifford reductions \cite{campbell_aux}. 

\begin{definition}[Clifford Reductions]
  An n-to-1 Clifford reduction takes an n copies of a state $\rho$ and outputs a single qubit $\rho'$ using ideal Clifford unitaries, preparation of stabilizer states, classical feedforward, classical randomness, Pauli measurements, and postselection. 
\end{definition}

Here, classical feedforward refers to the ability of the device to change the course of the protocol based on outcomes of measurements. Similarly, postselection is the ability to only include outcomes of the protocol which meet certain measurement criterion. In this paper we determine a region of non-stabilizer states which cannot be exited using Clifford reductions. Since Clifford reductions have many degrees of freedom, it is simpler to analyze a more restricted class of reduction. 

\begin{definition}[Stabilizer Reduction]
  Let $\rho$ be a single qubit state and let $S \leq P_n$ be some stabilizer subgroup with $n - 1$ independent generators. Let $D: \Q_n \rightarrow \Q_n$ be a Clifford decoding for $S$. An n-to-1 stabilizer reduction performs the following: (i) Take $\rho^{\otimes n}$; (ii) measure the $n - 1$ independent generators of $S$; (iii) repeat steps (i) and (ii) until we measure "$+1$" on all generators; (iv) return $\rho'$ where $D\rho''D^\dagger = \rho' \otimes \proj{0^{n - 1}}$ and $\rho''$ is the post-measurement state. 
\end{definition}

Observe that step (iv) is well defined since the post-measurement state is a convex combination of states in $\stab{(S)}$. While at first stabilizer reductions may seem far less general then their Clifford counterparts, it turns out that a state is distillable via Clifford reduction only if it is distillable via stabilizer reduction \cite{campbell_aux}.  

\section{Main Result}

We are now ready to state our main theorem. 

\begin{theorem}
  For all single-qubit states $\rho$ with $B(\rho) > 0$ and for all n-to-1 stabilizer reductions which output $\rho'$, if
  \begin{equation}
    1 < F(\rho) < \frac{6^n2^{3n - 2} + m^n}{6^n2^{3n - 2} - m^n}
  \end{equation}
  then $\rho'$ is a convex combination of a stabilizer state and $C\rho C^\dagger$ for some $C \in \mathcal{C}_1$. Here, $m = \min\{B_X,B_Y,B_Z\} / F$. 
  \label{thm:stab}
\end{theorem}

Note that here we are defining $B(\rho) > 0$ to be comparison component wise. Using \cite{campbell_aux} we can write down a more general theorem.
\begin{theorem}
  For all single-qubit states $\rho$ and for all n-to-1 Clifford reductions which output $\rho'$, if
  \begin{equation}
    1 < F(\rho) < \frac{6^n2^{3n - 2} + m^n}{6^n2^{3n - 2} - m^n}
  \end{equation}
  then $\rho'$ is a convex combination of a stabilizer state and $C\rho C^\dagger$ for some $C \in \mathcal{C}_1$. Here $m = \min\{|B_X|,|B_Y|,|B_Z|\}/F$. 
\end{theorem}
We are able to drop the positive-octant requirement on $B(\rho)$ since if a state cannot be improved in the positive octant then all of its Clifford rotates into the other octants cannot be improved as well. The proof of Theorem \ref{thm:stab} will rely on expressing $\rho$ in terms of a basis of pure stabilizer states and then utilizing the structure of stabilizer subgroups. Before we begin its proof we will prove several useful lemmas.

\subsection{Useful Lemmas}

Let $g \in R^{n}_\pm$ be a tuple of Pauli matrices. We will define 
\begin{equation}
  \mathbf{g_i} := \underbrace{I \otimes \dots \otimes I}_{i - 1\text{ times}} \otimes g_i \otimes \underbrace{I \otimes \dots \otimes I}_{n - i \text{ times}}
\end{equation}
Using this we define $\mathcal{G}_g$ to be the group generated by $\langle \mathbf{g_1}, \dots, \mathbf{g_n} \rangle$. Finally, $\ket{\psi_g}$ will denote the unique state stabilized by $\mathcal{G}_g$. 

\begin{lemma}
  Let $\ket\psi \in \Q_n$ be a stabilizer state with stabilizer group $G$. Let $S \leq P_n$ be a rank $n - 1$ stabilizer subgroup. If $\Pi_S\ket\psi \neq 0$, then there exists $g \in G$, such that the set $\{g, s_1, \dots, s_{n - 1}\}$ is independent and stabilizes $\Pi_S\ket\psi$. 
  \label{lemma:logical}
\end{lemma} 
\begin{proof}
  Let $\{s_1, \dots, s_{n - 1}\}$ and $\{g_1,\dots, g_n\}$ independent generating sets for $S$ and $G$. Then,
  \begin{equation}
    \Pi_S = \frac{1}{2^{n - 1}}\prod_{i \in [n - 1]}(I + s_i)
  \end{equation}
  We will first show that $(I + s_1)\ket\psi$ is stabilized by the independent set $\{s_1, g_{1}', \dots , g_{n - 1}'\}$ for some $g_{i}' \in G$. Consider the case where $s_1$ commutes with all $g_i$. Then, for all $g \in G$, $sg = gs$ and so 
  \begin{equation}
    g(I + s_1)\ket{\psi} = (g + gs_1)\ket{\psi} = g\ket{\psi} + s_1g\ket{\psi} = \ket{\psi} + s_1\ket{\psi} = (I + s_1)\ket{\psi}
  \end{equation}
  and so $(I + s_1)\ket\psi$ is stabilized by all elements of $G$. Observe that, $(I + s_1)\ket{\psi}$ is in the $+1$ eigenspace of $s_1$ and so $s_1$ is also a stabilizer. Since the stabilizer group of an n-qubit state has at most $2^n$ elements and $|G| = 2^n$, $s_1 \in G$. Since we can always pick a set of $n$ generators for $G$ which includes $s$, we get the desired result. This follows by using the fact that a set of Pauli operators is independent iff their binary symplectic representations are independent, and so we can just use the basis extension theorem for vector spaces. \\
  \\
  Now consider the case where $s_1$ anti-commutes with $k > 0$ original generators. We can assume without loss of generality that $s_1$ anti-commutes with $g_1, \dots, g_k$. Observe that, $g_1g_i$ is a stabilizer for $(I + s_1)\ket{\psi}$ for all $1 < i \leq k$. Precisely,
  \begin{equation}
    g_1g_i(I + s_1)\ket{\psi} = g_1g_i\ket{\psi} + g_1g_is_1\ket{\psi} = \ket{\psi} + s_1g_1g_i\ket{\psi} = (I + s_1)\ket{\psi}
  \end{equation}
  Therefore, we can define the new set of generators $G'$ to be 
\begin{equation}
  G' = \{s_1\} \cup (g_i : k < i \leq n) \cup \{g_1g_i : 1 < i \leq k\}
\end{equation}
All that remains is to show that $G'$ is independent. Take $H \subset G'$. Assume for contradiction that 
\begin{equation}
  I = \prod_{g \in H} g'
\end{equation}
If $s_1 \in H$ then multiply both sides by $s_1$ and we immediately get a contradiction since we have found a subset of $G$ which multiplies to $s_1$ and $s_1 \notin G$. If $s_1 \notin H$ we can cancel out all duplicate terms on the RHS, and so we get a subset of independent generators of $G$ which multiply to $I$ which is again a contradiction.  \\

To get the actual result we want, we repeat this process with $s_2, \dots, s_{n - 1}$. Note that we will never remove an $s_i$ once we add it because all $s_i$ commute. 
\end{proof}

\begin{lemma}
  Let $g \in R^{n}$ and let $S \leq P_n$. If $\Pi_S\ket{\psi_g} \neq 0$ then $\bra{\psi_g}\Pi_S\ket{\psi_g} \geq \frac{1}{2^{3n - 2}}$. 
  \label{lemma:bound}
\end{lemma}
\begin{proof}
  $\ket{\psi_g}$ is some tensor product of $\ket{+},\ket{+i}$ and $\ket{0}$ and so we can write
\begin{equation}
  \left(\frac{1}{\sqrt{2}}\right)^n \sum_{z \in \{0,1\}^n}c_z\ket{z} = k\ket{\psi_g} \coloneqq \ket{\psi}
\end{equation}
for some $c_z \in \{\pm 1, \pm i, 0\}$ and $0 < k \leq 1$. Since $\Pi_S = \frac{1}{2^{n - 1}}\sum_{s \in S}s$, 
\begin{equation}
  P\ket{\psi} = \frac{1}{2^{n - 1}}\left(\frac{1}{\sqrt{2}}\right)^n \sum_{z \in \{0,1\}^n}\left(l_z + p_zi\right)\ket{z}
\end{equation}
for $l_z,p_z \in \Z$. Since $\Pi_S\ket{\psi} = \Pi_Sk\ket{\psi_g} \neq 0$, at least one $l_z + p_zi \neq 0$. Finally
\begin{equation}
  \bra{\psi}\Pi_S\ket{\psi} = \frac{1}{2^{3n - 2}}\sum_{z \in \{0,1\}^n} l^{2}_z + p^{2}_z \geq \frac{1}{2^{3n - 2}}
\end{equation}
Observe, $\bra{\psi}\Pi_S\ket{\psi} = k^2\bra{\psi_g}\Pi_S\ket{\psi_g}  \geq \frac{1}{2^{3n - 2}} \implies \bra{\psi_g}\Pi_S\ket{\psi_g} \geq \frac{1}{2^{3n - 2}}$
\end{proof}

\begin{lemma}
  Let $S \leq P_n$ be a stabilizer subgroup such that for all $i \in [n]$ there exists $s \in S$ with the ith tensor factor of $s$ not equal to the identity. Then, there exists $g,g' \in  R^{n}$ such that the following hold:
  \begin{enumerate}[(i)]
    \item There exists $\overline{Z} \in \mathcal{G}_g$ and $s \in S$ such that $\overline{Z} \in N(S)$ and $-s\overline{Z} \in \mathcal{G}_{g'}$.
    \item $\Pi_S\ket{\psi_g} \neq 0$ and $\Pi_S\ket{\psi_{g'}} \neq 0$. 
  \end{enumerate}
  \label{lemma:campbell}
\end{lemma}

The proof of this lemma is the main result in Campbell's paper used to prove the existence of bound magic states outside the stabilizer octahedron \cite{CampbellBrowne2010}. 

\subsection{Proof of Theorem \ref{thm:stab}}
Since every n-to-1 stabilizer reduction is completely described by a state, n, a stabilizer subgroup, and a Clifford decoding, take $\rho \in \mathcal{S}(\Q_1)$ (assuming the conditions on $B(\rho)$ and $F(\rho)$ given by Theorem \ref{thm:stab}), $n \in \N$ and $S \leq P_n$ to be any stabilizer subgroup. Note that we will use a particular choice of Clifford decoding (which will be constructed from $S$), but this is completely general since all such decoding maps differ by some $C \in \mathcal{C}_1$ on the first qubit and as we stated earlier, Cliffords preserve the L1 Bloch norm. 

Moreover, we will also assume that for all $i \in [n]$ there is some $s \in S$ such that the ith tensor factor of $s$ is not the identity. If this were not the case that the post-measurement state is simply some copies of $\rho$ padded with stabilizer states. Then, the decoding will simply output some Clifford rotation of $\rho$. 

Let $A_{\pm L} = (1 \pm F)B_L/2F$ for $L \in R$ and let $m = \min\{B_X,B_Y,B_Z\}/F$. Observe that we can write 
\begin{equation}
  I = \frac{B_X}{F} (\proj{+} + \proj{-}) + \frac{B_Y}{F} (\proj{+i} + \proj{-i}) + \frac{B_Z}{F} (\proj{0} + \proj{1})
\end{equation}
since each term is a resolution of the identity and all the coefficients sum to 1. Plugging this into definition \ref{def:bloch} and using the spectral decomposition of each Pauli matrix we get
\begin{equation}
  \rho = A_X \proj{+} + A_Y\proj{+i} + A_Z \proj{0} + A_{-X}\proj{-} + A_{-Y} \proj{-i} + A_{-Z} \proj{1}
\end{equation}
and so we can write
\begin{equation}
  \rho^{\otimes n} = \sum_{g \in R^{n}} q_g \proj{\psi_g} + \sum_{g \in R^{n}_\pm \setminus R^{n}} p_g \proj{\psi_g}
\end{equation}
where $q_g = p_g = \prod_{i \in [n]} A_{g_i}$. Note that for all coefficients, $q_g \geq (1 + F)^n m^n/2^n$ and $p_g \geq (1 + F)^{n - 1}(1 - F)/2^n$. It will be important later that $p_g$ is allowed to be negative while $q_g$ is not. 

We will now compute step (ii) of the stabilizer reduction. We define
\begin{equation}
  \ket{\psi_{g}'} = 
  \begin{cases}
    \frac{\Pi_S\ket{\psi_g}}{\sqrt{\bra{\psi_g}\Pi_S\ket{\psi_g}}} & \text{If $\Pi_S\ket{\psi_g} \neq 0$} \\ 
    0 & \text{Otherwise}
  \end{cases}, \qquad \text{and} \qquad q_{g}' = p_{g}' = \bra{\psi_{g}}\Pi_S\ket{\psi_g}q_{g}
  \label{eq:q'}
\end{equation}
Note that $q_{g}'$ and $p_{g}'$ are real with $0 \leq q_g' \leq q_g$ and $|p_g'| \leq |p_g|$. Then,
\begin{equation}
  \Pi_S\rho^{\otimes n}\Pi_S = \sum_{g \in R^n} q_g' \proj{\psi_g'} + \sum_{g \in R_{\pm}^n \setminus R^n} p_g' \proj{\psi_g'}
  \label{eq:post-proj}
\end{equation}
From equation \eqref{eq:q'}, using Lemma \ref{lemma:bound} and the earlier bound on each $q_g$, we have that if $q_g' \neq 0$ then 
\begin{equation}
  q_g' \geq \left(\frac{1 + F}{2}\right)^n\frac{m^n}{2^{3n - 2}}
\end{equation}
Now by Lemma \ref{lemma:logical} we have that each non-vanishing $\ket{\psi_{g}'}$ is stabilized by a group with independent generators $\{G_g, s_1, \dots, s_{n - 1}\}$ where the $s_i$ operators form a generating set for $S$ and $G_g \in \mathcal{G}_g$. Recall that $\mathcal{G}_g$ is the group stabilizing $\ket{\psi_g}$. Since $G_g$ commutes with all $S$, it is a logical operator for the code. We want to show that in particular, $G_{g}^2 = I$ and $G_g \notin \pm S$. To show the former, simply observe that $\mathcal{G}_g$ is a stabilizer subgroup and so none of its elements can have phase $\pm i$. To see the latter, $G_g \notin S$ since otherwise $\{G_g, s_1, \dots, s_{n - 1}\}$ is not independent and $G_g \notin -S$ because then by closure we have $-I$ in the stabilizer group of $\ket{\psi_{g}'}$ which is a contradiction. Therefore, we have that $G_g$ is in one of $\pm X_L,\pm Y_L,\pm Z_L$ where each of these are cosets of $S$ in $N(S)$. Note that by labeling the cosets we have implicitly chosen an isomorphism $N(S) / S \cong P_1$, but this is completely general since the choice of isomorphism simply swaps the labels between these cosets. 

By Lemma \ref{lemma:campbell} we have two non-vanishing $\ket{\psi_{g}'}$ and $\ket{\psi_{g'}'}$ such that the former is stabilized by $\overline{Z} \in \mathcal{G}_g$ and the latter by $-s\overline{Z} \in \mathcal{G}_{g'}$. By Lemma \ref{lemma:encoding_logical} there exists a Clifford decoding with respect to $\overline{Z}$ as its logical $Z$ operator and so we will use this one. 

Since we know that there are only 6 choices of $G_g$ up to stabilizer, we know that there are only 6 distinct states in equation \eqref{eq:post-proj}
and so we can collection them. 
\begin{equation}
  \Pi_S\rho^{\otimes n}\Pi_S = \sum_{L \in R_\pm} q_L \proj{\psi_L} + \sum_{L \in R_\pm} p_L \proj{\psi_L}
\end{equation}
Here, $\ket{\psi_L}$ the unique $+1$ logical eigenstate of logical operator $L$ and both $q_L$ and $p_L$ are summed over $q_{g}'$ and $p_{g}'$ with $G_g \in L$. 

\begin{lemma}
  For all subsets $A \subseteq R_\pm$ 
  \begin{equation}
    q_{\pm Z} + \sum_{L \in A} p_L > 0
  \end{equation}
  \label{lemma:sum}
\end{lemma}
\begin{proof}
  Since $q_Z \geq q_g \neq 0$ and $q_{-Z} \geq q_{g'} \neq 0$ we have 
  \begin{equation}
    q_{\pm Z} \geq \left(\frac{1 + F}{2}\right)^n \frac{m^n}{2^{3n - 2}}
  \end{equation}
   Since there are at most $6^n$ terms in the second sum of equation \ref{eq:post-proj} and each $p_g \geq (1 + F)^{n - 1}(1 - F)/2^n$ we have 
  \begin{equation}
\sum_{L \in A} p_L \geq 6^n\left(\frac{1 + F}{2}\right)^{n - 1}\left(\frac{1 - F}{2}\right)
  \end{equation}
  and so 
  \begin{align}
    q_{\pm Z} + \sum_{L \in A} p_L &\geq 6^n\left(\frac{1 + F}{2}\right)^{n - 1}\left(\frac{1 - F}{2}\right) + \left(\frac{1 + F}{2}\right)^n\frac{m^n}{2^{3n - 2}}  \\
    &= \left(\frac{1 + F}{2}\right)^{n - 1}\left(\frac{1 - F}{2} + \frac{1 + F}{2}\frac{m^n}{2^{3n - 2}}\right)
  \end{align}
  Since $F < \frac{m^n + 6^n2^{3n - 2}}{6^n2^{3n - 2} - m^n}$ we get that the second term in the product is positive. The first term is clearly positive since $F \geq 0$ and so we get the desired result. 
\end{proof}

Finally, we want to renormalize the state in equation \eqref{eq:post-proj} and show the decoding is within the stabilizer octahedron. Let $c = \tr{\Pi_S\rho^{\otimes n}\Pi_S} = \sum q_L + p_L$. By lemma \ref{lemma:sum}, $c > 0$. Let $B'$ be the Bloch vector of the decoded state. Since the logical states are taken to their single qubit operator counterparts, we have that for $L \in R$ we can write
\begin{equation}
  B_{L}' = \frac{1}{c}((q_L + p_L) - (q_{-L} + p_{-L}))
\end{equation}
We want to show that 
\begin{equation}
  |B_{X}'| + |B_{Y}'| + |B_{Z}'| \leq 1
\end{equation}
which is equivalent to
\begin{equation}
  \sum_{L \in R} |(q_L + p_L) - (q_{-L} + p_{-L})| \leq c = \sum_{L \in R} q_L + p_L
\end{equation}
Define $M = \min\{q_Z + p_Z, q_{-Z} + p_{-Z}\}$. We can use the fact that $(q_{Z} + p_{Z}) \geq 0$ and $(q_{-Z} + p_{-Z}) \geq 0$ and the triangle inequality to write
\begin{equation}
  \sum_{L \in R} |(q_L + p_L) - (q_{-L} + p_{-L})| \leq -2M + p_Z + p_{-Z} + \sum_{L \in \pm\{Y, Z\}} |p_L| + \sum_{L \in R_\pm} q_L
\end{equation}
Let $A \subseteq \pm\{Y,Z\}$ label the subset of terms with $p_L < 0$. It then suffices to show that 
\begin{equation}
  -2M + \sum_{L \in A} |p_L|  = -2M + \sum_{L \in A} -p_L \leq \sum_{L \in A} p_L
\end{equation}
Now assume the contrary. Without loss of generality assume that $q_Z + p_Z \leq q_{-Z} + p_{-Z}$. 
\begin{align}
  -2q_{Z} - 2p_Z + \sum_{L \in A} -p_L &> \sum_{L \in A} p_L \\ 
  -2q_{Z} - 2p_Z  &> 2\sum_{L \in A} p_L \\ 
  -q_{Z} - p_Z  &> \sum_{L \in A} p_L \\ 
  0  &> q_{Z} + p_Z + \sum_{L \in A} p_L \\ 
\end{align}
the so we contradict lemma \ref{lemma:sum}. 

\section{Acknowledgements}

I wish to thank Prof. Tang for his guidance and help in composing this paper. I would also like to thank Thomas Bond and Scott McCary for their collaboration and Prof. Hoang for his guidance in doing mathematical research. 

\printbibliography{}
\onecolumn
\appendix

\section{Results on Stabilizer Groups}
\label{app:stab}
In this section we will prove some standard results about stabilizer states which are used throughout. 

\begin{lemma}
  Let $S \leq P_n$ be a stabilizer subgroup with independent generators $\{s_1, \dots, s_{n - k}\}$. Then the projection onto $\stab{(S)}$ is given by 
  \begin{equation}
    P_S = \frac{1}{2^{n - k}}\prod_{i \in [n - k]}(I + s_i) = \frac{1}{2^{n - k}}\sum_{s \in S} s
  \end{equation}
\end{lemma}

\begin{lemma}
  Let $S \leq P_n$ be a rank $n - k$ stabilizer subgroup. Then, $\dim{(\stab{(S)})} = 2^k$.
\end{lemma}
\begin{proof}
  The projector onto the simultaneous $+1$ eigenspace is
  \begin{equation}
    \Pi_S = \frac{1}{2^{n - k}}\sum_{s \in S} s
  \end{equation}
  Since $\Pi_S$ is an orthogonal projection it can only have eigenvalues $1$ or $0$. Let $a_0,a_1$ and $g_0,g_1$ be the algebraic and geometric multiplicities of $0$ and $1$ respectively. Since $\Pi_S$ is Hermitian, $a_0 = g_0$ and so $\dim{(\nullspace{\Pi_S})} = a_0$. The rank-nullity theorem then gives $a_0 + \rank{(\Pi_S)} = n$ and since $a_0 + a_1 = n$,
  \begin{equation}
    \rank{(\Pi_S)} = a_1 = \tr{(\Pi_S)}
  \end{equation}
  Expanding,
  \begin{equation}
    \tr{(\Pi_S)} = \frac{1}{2^{n - k}}\left(\tr{I} + \sum_{\substack{s \in S \\ s \neq I}} \tr{s}\right) = \frac{1}{2^{n - k}}\left(2^n + \sum_{\substack{s \in S \\ s \neq I}} \prod_{i \in [n]} \tr{s_i}\right) = \frac{2^n}{2^{n - k}} = 2^k
  \end{equation}
  where $s_i$ denotes the ith tensor factor of $s$. If $s \in S$ does not equal $I$ then $\exists i \in [n]$ such that $s_i \in \{X, Y, Z\}$ (since $\pm i I, -I \notin S$). Each of these are traceless and so $\tr{s_i} = 0$ and so the product vanishes. Finally,
  \begin{equation}
    \dim{(\im{\Pi_S}}) = \rank{(\Pi_S)} = 2^k
  \end{equation}
\end{proof}

\begin{lemma}
Let $\ket{\psi}$ be an n-qubit stabilizer state with stabilizer group $G$. Let $S$ be any stabilizer group. Then, $\Pi_S\ket{\psi} = 0$ if and only if there exists $s \in S$ such that $-s \in G$. 
\end{lemma}
\begin{proof}
  Let $S'$ be a set of $n$ independent generators for $S$. Assume that $\exists s \in S$ such that $-s \in G$. Then, 
  \begin{equation}
    s\ket{\psi} = s(-s)\ket{\psi} = -s^2\ket{\psi} = -\ket{\psi}
  \end{equation}
  and so $\ket{\psi}$ is a $-1$ eigenvector of $s$. Note that $s^2 = I$ because $s$ cannot have phase $\pm i$. Since the stabilizer group commutes,
  \begin{equation}
    \Pi_S = \frac{1}{2^n}\prod_{x \in S'}(I + x) = \frac{1}{2^n}\left(\prod_{\substack{x \in S' \\ x \neq s}} (I + x)\right) (I + s)
  \end{equation}
  and so 
  \begin{equation}
    \Pi_S\ket{\psi} = \frac{1}{2^n}\left(\prod_{\substack{x \in S' \\ x \neq s}} (I + x)\right)(I + s)\ket{\psi} = \frac{1}{2^n} \left(\prod_{\substack{x \in S' \\ x \neq s}} (I + x)\right)(\ket{\psi} - \ket{\psi}) = 0
  \end{equation}
  To prove the other direction, assume that $\Pi_S\ket{\psi} = 0$. Since $\ket{\psi}$ is a stabilizer state, $G$ stabilizes a 1-dimensional subspace and so for all $\ket\phi \in \C^{2^n}$, $\Pi_G\ket\phi = c\ket\psi$ for some $c \in \C$. Therefore, $\Pi_S\Pi_G\ket{\phi} = \Pi_Sc\ket\psi = 0$. Since this holds for all $\ket\phi$, $\Pi_S\Pi_G$ is the zero map and in particular $\tr{(\Pi_S\Pi_G)} = 0$. Expanding,
  \begin{equation}
    \tr{(\Pi_S\Pi_G)} = \tr{\left(\frac{1}{2^{2n}} \sum_{s,g \in S \times G} sg\right)} = \frac{1}{2^{2n}}\left(\tr{I} + \sum_{\substack{s,g \in S \times G \\ s,g \neq (I,I)}}{\tr{sg}}\right) = 0
    \label{eq:tmp}
  \end{equation}
  Note that since $\tr{I} = 2^n$ the summation term must be negative. For all $s,g \in S \times G$ we can write
  \begin{equation}
    s = \eta_sP_1 \otimes \dots \otimes P_n, \quad g = \eta_gQ_1 \otimes \dots \otimes Q_n, \qquad \eta_s,\eta_g \in \{-1, 1\}, P_i,Q_i \in \{X, Y, Z\}
  \end{equation}
  Therefore, $sg \neq \pm i I$. If $sg \neq \pm I$ then $\exists i \in [n]$ such that $P_iQ_i \in \{X, Y, Z\}$ and so $\tr{sg} = \tr{(P_iQ_i)}\prod_{j \neq i} \tr{P_jQ_j} = 0$. Therefore, $\tr{(sg)} \in \{-2^n,2^n, 0\}$. Since the summation term in equation \eqref{eq:tmp} must be negative there exists $s,g \in S \times G$ such that $sg = -I$. Observe that $sg = -I$ only if for all $i \in [n], P_i = Q_i$ and $\eta_s = -\eta_g$ and so $s = -g$. 
\end{proof}

\end{document}